\documentclass[11pt,oneside]{amsart}
\usepackage{amsaddr}

\usepackage[utf8]{inputenc}
\usepackage[T1]{fontenc}
\usepackage{microtype}

\usepackage[margin=1in]{geometry}

\usepackage{amsmath}
\usepackage{amssymb}
\usepackage{amsthm}
\usepackage{mathtools}
\usepackage{thmtools}

\usepackage[style=trad-alpha,
            natbib=true]{biblatex}

\usepackage[tt=false]{libertine}
\usepackage[libertine]{newtxmath}

\usepackage{array}
\usepackage{booktabs}
\usepackage[inline,shortlabels]{enumitem}
\usepackage{graphicx}
\usepackage{xcolor}

\usepackage{hyperref}
\usepackage{cleveref}

\hypersetup{
  linkcolor  = blue!80!black,
  citecolor  = blue!80!black,
  colorlinks = true
}

\theoremstyle{plain}
\newtheorem{theorem}{Theorem}[section]
\newtheorem{corollary}[theorem]{Corollary}
\newtheorem{proposition}[theorem]{Proposition}
\newtheorem{lemma}[theorem]{Lemma}

\theoremstyle{remark}
\newtheorem*{remark*}{Remark}

\DeclarePairedDelimiter{\paren}{\lparen}{\rparen}

\let\temp\phi
\let\phi\varphi
\let\varphi\temp

\makeatletter
\let\@afterindenttrue\@afterindentfalse
\@afterindentfalse
\makeatother

\makeatletter
\newcommand{\todo}[1]{\textcolor{red}{[TODO\@ifnotempty{#1}{: #1}]}}
\makeatother

\newcommand{\cA}{\mathcal A}
\newcommand{\cB}{\mathcal B}
\newcommand{\cC}{\mathcal C}
\newcommand{\eps}{\varepsilon}

\newcommand{\PredM}{\textsc{PredictiveMarker}}
\newcommand{\Blind}{\textsc{BlindOracle}}
\newcommand{\LNM}{\textsc{LNonMarker}}
\newcommand{\M}{\textsc{Marker}}
\newcommand{\FtP}{\textsc{FollowThePrediction}}
\newcommand{\OPT}{\mathsf{OPT}}
\newcommand{\ALG}{\mathsf{ALG}}

\title{Better and Simpler Learning-Augmented Online Caching}
\author{Alexander Wei}
\address{Harvard University}
\email{weia@college.harvard.edu}

\bibliography{caching}

\begin{document}

\maketitle

\begin{abstract}
  Lykouris and Vassilvitskii (ICML 2018) introduce a model of online caching with machine-learned advice, where each page request additionally comes with a prediction of when that page will next be requested. In this model, a natural goal is to design algorithms that (1) perform well when the advice is accurate and (2) remain robust in the worst case a la traditional competitive analysis. Lykouris and Vassilvitskii give such an algorithm by adapting the \textsc{Marker} algorithm to the learning-augmented setting. In a recent work, Rohatgi (SODA 2020) improves on their result with an approach also inspired by randomized marking. We continue the study of this problem, but with a somewhat different approach: We consider combining the \textsc{BlindOracle} algorithm, which just na\"ively follows the predictions, with an optimal competitive algorithm for online caching in a black-box manner. The resulting algorithm outperforms all existing approaches while being significantly simpler. Moreover, we show that combining \textsc{BlindOracle} with \textsc{LRU} is in fact optimal among deterministic algorithms for this problem.
\end{abstract}

\section{Introduction}\label{sec:intro}

Traditionally, the study of online algorithms focuses on \emph{worst-case} robustness, with algorithms providing the same competitive guarantee against the offline optimal over all inputs. In recent years, however, there has been a surge of interest in studying online algorithms in the presence of structured inputs
\cite{MV17, LV18, PSK18, GP19, KPSSV19, KPSV20, Rohatgi20, LLMV20, Mitzenmacher20}. A principal motivation for these works is the philosophy of beyond worst-case analysis \cite{KP00, Roughgarden19}: Many practical settings have inputs that follow restricted patterns, making classical worst-case competitive analysis too pessimistic to inform practice. In particular, the worst-case examples classical competitive analyses guard against often do not materialize.  Furthermore, algorithms designed with the worst case in mind can be hamstrung by these considerations, ending up unnatural and losing performance on ``nice'' inputs.

\emph{Learning-augmented} online algorithms, introduced by \citet{LV18, PSK18}, is a beyond worst-case framework motivated by the powerful predictive abilities of modern machine learning. The structure of the input is assumed to come in the form of a machine-learned predictor that provides predictions of future inputs. A concern with this setup may be that machine learning models typically have few worst-case guarantees. Nonetheless, with learning-augmented algorithms, we want the best of both worlds: Given a predictor, the objective is to design algorithms that
\begin{enumerate*}
  \item perform well in the optimistic scenario, where the predictor has low error, and
  \item remain robust in the classical worst-case sense, when the predictor can be arbitrarily bad.
\end{enumerate*}
That is, we would like our algorithm to be $c(\eta)$-competitive against the offline optimal on all inputs, where $c$ is a function of the predictor error $\eta$ such that $\max_\eta c(\eta)\le\gamma$ for some $\gamma\ge 1$. Such an algorithm is said to be \emph{$\gamma$-robust}.

The focus of our work is learning-augmented online caching \cite{LV18, Rohatgi20}. In the \emph{online caching} (a.k.a. \emph{online paging}) problem, one seeks to maintain a cache of size $k$ online while serving requests for pages that may or may not be in the cache. For simplicity, we assume that pages must always be served from the cache and that bringing a page into the cache has unit cost. (In particular, if the cache is full, bringing a page into the cache requires also \emph{evicting} a page already in the cache.) Thus, we seek to minimize the number of requests for which the page is not already in the cache, i.e., the number of \emph{cache misses}. This is a classical online problem that has been the subject of extensive study over the past several decades (see \cite{BE98} for an overview). From the worst-case perspective, this problem is well-understood for not only the version stated above \cite{ST85, FKLMSY91, ACN00}, but also for weighted generalizations \cite{BBN12, BBN12b}.

Online caching in the learning-augmented context was first considered by \citet{LV18}. They introduce a model of prediction where the predictor, upon the arrival of each page, predicts the next time that this page will be requested. They show that the \Blind{} algorithm, which follows the predictor na\"ively and evicts the page with the latest predicted arrival time, can have unbounded competitive ratio (i.e., is \emph{non-robust}). They then give a different algorithm, \PredM{}, based on the \M{} algorithm of \citet{FKLMSY91}, that achieves a competitive ratio of
\[ 2 + O\paren*{\min\paren*{\sqrt{\frac{\eta}{\OPT}}, \log k}}, \]
where $\eta$ is the $\ell_1$ error of the predictor and $\OPT$ is the cost of the offline optimal. (See \Cref{sec:setup} for precise definitions.) In \cite{Rohatgi20}, Rohatgi introduces the \LNM{} algorithm, which is also based on randomized marking (but eschews the framework somewhat), and shows that this algorithm obtains a competitive ratio of
\[ O\paren*{1 + \min\paren*{\frac{\log k}{k}\frac{\eta}{\OPT}, \log k}}. \]
This bound is obtained by first constructing a non-robust algorithm and then using a black-box combination technique discussed in \cite{LV18} to combine this non-robust algorithm with the \M{} algorithm. Rohatgi also provides a lower bound of
\[ \Omega\paren*{\min\paren*{\log\paren*{\frac{1}{k\log k}\frac{\eta}{\OPT}}, \log k}} \]
for the competitive ratio of any learning-augmented online algorithm for caching in terms of $k$, $\OPT$, and $\eta$.

\subsection{Our Contribution}

We show that the strikingly simple approach of combining \Blind{} with an $O(\log k)$-competitive online caching algorithm (e.g., \M{}) in a black-box fashion obtains a competitive ratio bound of
\[ O\paren*{1 + \min\paren*{\frac 1k\frac{\eta}{\OPT}, \log k}}, \]
improving on that of \LNM{}. Thus, although \Blind{} is non-robust \cite{LV18}, we show that it should not be abandoned entirely. In fact, it has excellent performance when $\eta/\OPT$ is small. So when this algorithm is combined with a $O(\log k)$-competitive algorithm, we start seeing performance improvement over the robust algorithm starting at $\eta/\OPT = O(k\log k)$.

And although our improvement in the competitive ratio is slight, previous approaches for learning-augmented online caching \cite{LV18, Rohatgi20} have relied on much more intricate constructions based on randomized marking. We therefore believe that our simple approach may yield better practical performance and may generalize more readily to other learning-augmented settings. (Indeed, we note that the deterministic combining algorithm is particularly simple: Just ``follow'' the algorithm with the better performance thus far.)

We also give precise bounds on the constant factors in the competitive ratios that we obtain. \citet{FKLMSY91,BB00} provide optimal bounds for combining online algorithms online in a black-box manner, with better constants than the approach discussed in \cite{LV18} and applied in \cite{Rohatgi20}. By composing a careful competitive analysis of \Blind{} with these ``combiners,'' we obtain constants in the competitive ratio that are lower than those of previous work.

Finally, we show that combining \Blind{} with a $k$-competitive deterministic algorithm (e.g., LRU \cite{ST85}) is the best one could hope to do among deterministic algorithms for learning-augmented online caching. In particular, we show that a linear dependence on $\eta/(k\cdot\OPT)$ in the competitive ratio is necessary. Therefore, if a logarithmic dependence on $\eta/(k\cdot\OPT)$ is to be achieved, as in Rohatgi's lower bound, then randomization is needed, (perhaps surprisingly) even in the regime where $\eta/(k\cdot\OPT)$ is bounded.

Stated formally, our main result analyzing \Blind{} is the following:

\begin{theorem}[restate=maintheorem,label=theorem:main]
  For learning-augmented online caching, \Blind{} obtains a competitive ratio of
  \[ \min\paren*{1 + 2\frac{\eta}{\OPT}, 2 + \frac{4}{k-1}\frac{\eta}{\OPT}}, \]
  where $\eta$ is the $\ell_1$ loss incurred by the predictor and $\OPT$ is the offline optimal cost. (For precise definitions, see \Cref{sec:setup}.)
\end{theorem}

Plugging this bound into the results of \citet{FKLMSY91, BB00} for competitively combining online algorithms online (see \Cref{sec:combine}) yields the following corollaries:

\begin{corollary}[restate=firstcorollary,label=corollary:det]
  There exists a deterministic algorithm for learning-augmented online caching that achieves a competitive ratio of
  \[ 2\min\paren*{\min\paren*{1 + 2\frac{\eta}{\OPT}, 2 + \frac{4}{k-1}\frac{\eta}{\OPT}}, k}. \]
\end{corollary}

\begin{corollary}[restate=secondcorollary,label=corollary:rand]
  There exists a randomized algorithm for learning-augmented online caching that achieves a competitive ratio of
  \[ (1 + \eps)\min\paren*{\min\paren*{1 + 2\frac{\eta}{\OPT}, 2 + \frac{4}{k-1}\frac{\eta}{\OPT}}, H_k} \]
  for any $\eps\in (0, 1/4)$.\footnote{The trade-off in $\eps$ and the additional cost is additive; thus, it does not factor into the competitive ratio.} (Here, $H_k = 1 + \frac 12 + \frac 13 + \cdots + \frac 1k = \ln(k) + O(1)$ is the $k$-th harmonic number.)
\end{corollary}

Finally, we state our matching lower bound for \Cref{corollary:det} on deterministic algorithms for learning-augmented online caching:

\begin{theorem}[restate=lowerboundtheorem,label=theorem:lower]
  The competitive ratio bound for any deterministic learning-augmented online caching algorithm must be at least
  \[ 1 + \Omega\paren*{\min\paren*{\frac{1}{k}\frac{\eta}{\OPT}, k}}. \]

\end{theorem}

\subsection{Related Work}

In addition to the predecessor works by \citet{LV18, Rohatgi20} on learning-augmented online caching, there have been several other recent papers in the space of learning-augmented online algorithms: \citet{MV17} study repeated posted-price auctions, \citet{PSK18, GP19} study the ski rental problem, and \citet{PSK18, LLMV20, Mitzenmacher20} study online scheduling. Of these, the scheduling algorithm of \citet{PSK18} is the most similar in spirit to this present work: Both algorithms are based on combining a na\"ive and optimistic algorithm with a robust algorithm.

Other threads of research falling under beyond worst-case online algorithms include work on combining multiple algorithms with different performance characteristics \cite{FKLMSY91, BB00, MNS12, GP19}, designing online algorithms with distributional assumptions (e.g., stochasticity) on the input \cite{KP00, BS12, MGZ12}, and semi-online algorithms, where the input is assumed to have a predictable offline component and an adversarial online component \cite{KPSSV19, KPSV20}.

The idea of learning-augmentation has also been explored in many other algorithmic and data structural settings in recent years. These include learned indices \cite{KBCDP18}, bloom filters \cite{Mitzenmacher18}, frequency estimation in streams \cite{HIKV19}, and nearest neighbor search \cite{DIRW19}, among others.

Finally, advice for online algorithms has also been considered with a more complexity theoretic spirit through the study of advice complexity of online algorithms; see the survey of \citet{BFKLM17}.

\subsection{Recent Developments}

Recently, in work done independently of and concurrently with this paper, \citet{ACEPS20} also study a \Blind{}-like algorithm, which they term \FtP{}, in the more general setting of learning-augmented metrical task systems; they also use the ``combiner'' of \citet{BB00} to make this algorithm robust. However, their prediction model, when specialized to online caching, is incomparable to that of \citet{LV18} (which we follow).\footnote{Namely, their algorithms expect predictions to be in a different form: They expect predictions to be cache states (i.e., the set of pages in the cache at time $t$) rather than next arrival times of pages. Moreover, there exist sequences of ``corresponding'' inputs for each of these two models such that the predictor error approaches infinity in one model while remaining constant in the other.} Thus, the theoretical results proved in these two models do not imply each other.

\section{Preliminaries}

\subsection{Setup and Notation}\label{sec:setup}

In the online caching problem, we receive a sequence $\sigma = (\sigma_1,\ldots,\sigma_n)$ of page requests online, and our goal is to serve these requests using a cache of size $k$ while minimizing cost. In this problem, pages must be served from the cache and can be served at no cost; however, evicting a page from the cache has unit cost.\footnote{Note that this is equivalent to the ``standard'' version, where each cache miss has unit cost, up to a constant of $k$.}

We will establish competitive bounds comparing the performance of two online caching algorithms $\mathcal A$ and $\mathcal B$. More precisely, we will show bounds of the form $$
  \mathsf{ALG}_{\mathcal B}(\sigma)\le\gamma\cdot\mathsf{ALG}_{\mathcal A}(\sigma) + O(1),
$$ where $\mathsf{ALG}_{\mathcal A}(\sigma)$ and $\mathsf{ALG}_{\mathcal B}(\sigma)$ are the costs of $\mathcal A$ and $\mathcal B$, respectively, as measured in number of evictions made while serving a sequence $\sigma$ of page requests. We will also compare our costs to the offline optimal algorithm $\mathsf{OPT}$, whose cost $\mathsf{OPT}(\sigma)$ is the minimum possible cost of serving request sequence $\sigma$. We will omit the argument $\sigma$ when the context is clear (i.e., just writing $\ALG_{\mathcal A}$ to represent $\ALG_{\mathcal A}(\sigma)$).

In our analysis, we use $A_t$ and $B_t$ to denote the cache states of $\mathcal A$ and $\mathcal B$, respectively, just before the $t$-th request. Formally, $A_t$ and $B_t$ are subsets of $\{1,\ldots,t-1\}$ of size at most $k$, containing for each cached page the index at which it was last served. That is, when serving the $t$-th request, we remove some old request index $t'$ from the cache and insert $t$. Thus, if $t'$ is such that $\sigma_{t} = \sigma_{t'}$, this operation is free; otherwise, it has unit cost. In the sequel, we will also refer to these indices $t$ as \emph{page requests}.

In the learning-augmented online caching problem, the $t$-th page request comes with a prediction $h_t$ for the next time page $\sigma_t$ is requested. That is, at the time of the $t$-th request, our algorithm receives the pair $(\sigma_t, h_t)$. Let $h = (h_1,\ldots,h_n)$ be the tuple of all $n$ predictions. To define a notion of loss, let $y_t$ denote for each $t$ the next time page $t$ is actually requested, with $y_t = n + 1$ if page $\sigma_t$ is never requested again. The $\ell_1$ loss is then defined to be $$
    \eta(\sigma, h) = \sum_t |h_t - y_t|.
$$ We will omit arguments to $\eta$ if the context is clear. Note that if $\eta(\sigma, h) = 0$, then the offline optimal can be obtained, as the optimal algorithm always evicts the page that is next requested furthest into the future.

In stating our bounds, the essential quantity is often $\eta/\OPT$. To make this clear, we take $\eps = \eta/\OPT$ and state our bounds in terms of $\eps$ in the sequel.

\subsubsection{Inversions}

Call a pair $(i, j)$ of page requests an \emph{inversion} if $y_i < y_j$ but $h_i\ge h_j$. Let $M(\sigma, h)$ denote the total number of inversions between the pair of sequences $\sigma$ and $h$. And as above, we will omit the arguments to $M$ when the context is clear.

\subsubsection{BlindOracle}

We now formally define the \Blind{} algorithm as follows: For each page request, if the requested page is already in the cache, do nothing. Otherwise, evict the page request $p$ whose predicted next arrival time $h_p$ is furthest away among all $p\in A_t$, with ties broken consistently (e.g., by always evicting the least recently used page with maximal $h_p$).

\subsection{Combining Online Algorithms Competitively}\label{sec:combine}

In this section, we state some classical bounds on competitively combining online algorithms, due to \citet{FKLMSY91} and \citet{BB00}. This type of ``black-box'' combination was also considered by \citet{LV18}, but their approach has a worst constant than that of \citet{FKLMSY91}. We also note that results of a similar flavor are proven by \citet{PSK18, MNS12}, but for other online problems.

The question of combining multiple online algorithms while remaining competitive against each was first considered in the seminal paper of \citet{FKLMSY91}. They consider combining $n$ online algorithms $\cA_1,\ldots,\cA_n$ for the online caching problem into a single algorithm $\cA$ such that $\cA$ is $C_i$-competitive against $\cA_i$ for each $i$. They show that such an $\cA$ is achievable if and only if
\[ \sum_{i=1}^n \frac 1{C_i}\le 1. \]

We will need only the special case of $n = 2$ and $C_1 = C_2 = 2$, which we state below:

\begin{theorem}[\cite{FKLMSY91}, special case]\label{theorem:det_comb}
  Given any two algorithms $\cA$ and $\cB$ for the online caching problem, there exists an algorithm $\cC$ such that
  \[ \mathsf{ALG}_{\cC}(\sigma)\le 2\min(\mathsf{ALG}_{\cA}(\sigma), \mathsf{ALG}_{\cB}(\sigma)) + O(1). \]
  Moreover, if $\cA$ and $\cB$ are deterministic, so is $\cC$.
\end{theorem}

Indeed, we note that that this can be done deterministically with a ``follow-the-leader'' approach, in which we simulate both algorithms and at each step evict any page that is not in the cache of the better performing algorithm (as measured by total number of evictions after serving the current request).

\citet{BB00} show that one can obtain a better approximation factor using a \emph{randomized} scheme, namely multiplicative weights.\footnote{The result of \citet{BB00} in fact holds for general metrical task systems.} That is, at each point in time, the probability that the combined algorithm is following one of the $n$ algorithms is given by a probability distribution over the $n$ algorithms governed by the multiplicative weights update rule. For $n = 2$, their result can be stated as follows:

\begin{theorem}[\cite{BB00}, special case]\label{theorem:rand_comb}
  Given any two algorithms $\cA$ and $\cB$ for the online caching problem and any $\eps$, $0 < \eps < 1/4$, there exists an algorithm $\cC$ such that
  \[ \mathsf{ALG}_{\cC}(\sigma)\le (1 + \eps)\min(\mathsf{ALG}_{\cA}(\sigma), \mathsf{ALG}_{\cB}(\sigma)) + O(\eps^{-1}k). \]
\end{theorem}

\begin{remark*}
  Although we do not state the versions of these results for $n > 2$, one could imagine that they can be useful wishes to combine multiple machine-learned predictors.
\end{remark*}

\subsection{From $\ell_1$ Loss to Inversions}

We now state a lemma of \citet{Rohatgi20} that relates $\ell_1$ loss to the number of inversions, letting us lower bound the $\ell_1$ loss $\eta(\sigma, h)$ by lower bounding the number of inversions $M(\sigma, h)$. Thus, instead of reasoning in terms of $\ell_1$ loss, we will reason in terms of inversions.

\begin{lemma}[\cite{Rohatgi20}]\label{lemma:inv}
  For any $\sigma$ and $h$, $\eta(\sigma, h)\ge \frac 12 M(\sigma, h)$.
\end{lemma}

With this lemma, it suffices (up to a factor of $2$) to give our competitive ratio upper bounds in terms of the number of inversions $M$.

\section{A First Analysis of \Blind{}}

In this section, we give a first analysis of \Blind{}, showing that it gets very good performance when the ratio $\eps = \eta / \mathsf{OPT}$ is very small. In particular, our analysis shows that as $\eps\to 0$, the competitive ratio achieved approaches $1$.

Let $\mathcal A$ be the offline optimal algorithm (i.e., such that $\mathsf{ALG}_{\cA} = \OPT$). Let $\mathcal B$ be \Blind{}. Note that we can think of each of $\ALG_\cA$, $\ALG_\cB$, and $M$ as functions of the time $t$, i.e., they are the cost of $\cA$, the cost of $\cB$, and the number of inversions, respectively, on the prefix consisting of the first $t - 1$ requests.\footnote{This indexing is to be consistent with the definitions of $A_t$ and $B_t$.} We use the $\Delta$ operator to denote the change (in a function of $t$) from time $t$ to time $t+1$. For example, $\Delta\ALG_{\cA} = 1$ if $\ALG_\cA$ evicts an element upon the $t$-th request.

In our analysis, we maintain a matching $X_t$ between $A_t$ and $B_t$ at all times $t$. Call a matching \emph{valid} if it consists only of pairs $(a, b)\in A_t\times B_t$ such that the next arrival of $b$ is no later than the next arrival of $a$. Indeed, our matching $X_t\subseteq A_t\times B_t$ will be valid throughout the execution of the algorithm.

We now proceed with a potential function analysis, taking our potential $\Phi$ (as a function of $A_t$, $B_t$, and $X_t$) to be the number of unmatched pages in $B_t$. For notational simplicity, we will simply denote $\Phi(A_t, B_t, X_t)$ by $\Phi(t)$. Given this setup, we show:

\begin{proposition}\label{proposition:first}
There exists a valid matching $X_n$ such that $$
\mathsf{ALG}_{\mathcal B} + \Phi(n) \le \mathsf{OPT} + M.
$$
\end{proposition}

\begin{proof}
We induct on the length $n$ of the input and perform a case analysis to show that we can maintain a valid matching $X_t$ such that at each time step, the right-hand side increases at least as much as the left-hand side, i.e., $\Delta\ALG_\cB + \Delta\Phi\le\Delta\OPT + \Delta M$.

For our base case, note that $A_1 = B_1$, so we may take $X_1$ to be the identity matching.

Now, upon a request at time $t$, we update $X_t$ according to the following cases (and with the consequences listed for each case):
  \begin{enumerate}
    \item
      The requested page $p$ is in both $A_t$ and $B_t$.
      \begin{enumerate}
        \item The cached pages are matched to each other.
          \begin{itemize}
            \item Do nothing.
          \end{itemize}
        \item Otherwise:
          \begin{enumerate}
            \item Both cached pages are matched.
              \begin{itemize}
                \item Remove the pairs $(c, p)$ and $(p, d)$ from $X_t$.
                \item Add the pairs $(p, p)$ and $(c, d)$ to $X_t$.
                \item As a result:
                  \begin{itemize}
                    \item
                      $\Delta\Phi = 0$.
                  \end{itemize}
              \end{itemize}
            \item Otherwise:
              \begin{itemize}
                \item Remove any pairs involving $p$ from $X_t$. (There is at most one such pair.)
                \item Add the pair $(p, p)$ to $X_t$.
                \item As a result:
                  \begin{itemize}
                    \item
                      $\Delta\Phi\le 0$.
                  \end{itemize}
              \end{itemize}
          \end{enumerate}
      \end{enumerate}
        \item
          The requested page $p$ is in $B_t$ only.
          \begin{itemize}
            \item Remove any pairs involving the evicted page $a$ from $X_t$. (There is at most one such pair.)
            \item Remove any pairs involving the requested page $p$ from $X_t$. (There is at most one such pair.)
            \item Add the pair $(p, p)$ to $X_t$.
            \item As a result:
              \begin{itemize}
                \item
                  $\Delta\OPT = 1$.
                \item
                  $\Delta\Phi\le 1$.
              \end{itemize}
          \end{itemize}
    \item The requested page $p$ is in $A_t$ only.
      \begin{enumerate}
        \item The evicted page $b\in B_t$ is unmatched.
          \begin{itemize}
            \item Add the pair $(p, p)$ to $X_t$. (The arriving page $p\in A_t$ cannot be in any valid matching.)
            \item As a result:
              \begin{itemize}
                \item
                  $\Delta\ALG_\cB = 1$.
                \item
                  $\Delta\Phi = -1$.
              \end{itemize}
          \end{itemize}
        \item The evicted page $b\in B_t$ is matched.
          \begin{enumerate}
            \item $b$ arrives later than all unmatched pages in $B_t$.
              \begin{itemize}
                \item Remove the pair $(c, b)$ involving the evicted page $b$ from $X_t$.
                \item Add the pair $(c, b')$ to $X_t$, where $b'\in B_t$ is any unmatched page.
                \item Add the pair $(p, p)$ to $X_t$. (The arriving page $p\in A_t$ cannot be in any valid matching.)
                \item As a result:
                  \begin{itemize}
                    \item
                      $\Delta\ALG_\cB = 1$.
                    \item
                      $\Delta\Phi = -1$.
                  \end{itemize}
              \end{itemize}
            \item There is an unmatched page $b'\in B_t$ arriving later than $b$.
              \begin{itemize}
                \item Remove the pair $(c, b)$ involving the evicted page $b$ from $X_t$.
                \item Add the pair $(p, p)$ to $X_t$. (The arriving page $p\in A_t$ cannot be in any valid matching.)
                \item As a result:
                  \begin{itemize}
                    \item
                      $\Delta\ALG_\cB = 1$.
                    \item
                      $\Delta\Phi = 0$.
                    \item
                      $\Delta M = 1$, as there is an inversion between $b$ and $b'$. (Note that we do not count this inversion ever again, as $b$ gets evicted.)
                  \end{itemize}
              \end{itemize}
          \end{enumerate}
      \end{enumerate}
    \item The requested page $p$ is in neither $A_t$ nor $B_t$.
      \begin{enumerate}
        \item $\mathcal A$ evicts an unmatched page $a\in A_t$.
          \begin{enumerate}
            \item $\mathcal B$ evicts an unmatched page $b\in B_t$.
              \begin{itemize}
                \item Add the pair $(p, p)$ to $X_t$.
                \item As a result:
                  \begin{itemize}
                    \item
                      $\Delta\OPT = 1$.
                    \item
                      $\Delta\ALG_\cB = 1$.
                    \item
                      $\Delta\Phi = 1$.
                  \end{itemize}
              \end{itemize}
            \item $\mathcal B$ evicts a matched page $b\in B_t$.
              \begin{itemize}
                \item Remove the pair $(c, b)$ involving $b$ from $X_t$.
                \item Add the pair $(p, p)$ to $X_t$.
                \item As a result:
                  \begin{itemize}
                    \item
                      $\Delta\OPT = 1$.
                    \item
                      $\Delta\ALG_\cB = 1$.
                  \end{itemize}
              \end{itemize}
          \end{enumerate}
        \item $\mathcal A$ evicts a matched page $a\in A_t$.
          \begin{enumerate}
            \item $\mathcal B$ evicts an unmatched page $b\in B_t$.
              \begin{itemize}
                \item Remove the pair $(a, d)$ involving $a$ from $X_t$.
                \item Add the pair $(p, p)$ to $X_t$.
                \item As a result:
                  \begin{itemize}
                    \item
                      $\Delta\OPT = 1$.
                    \item
                      $\Delta\ALG_\cB = 1$.
                  \end{itemize}
              \end{itemize}
            \item $\mathcal B$ evicts a matched page $b\in B_t$.
              \begin{itemize}
                \item Remove the pair $(a, d)$ involving $a$ from $X_t$.
                \item Remove the pair $(c, b)$ involving $b$ from $X_t$.
                \item Add the pair $(p, p)$ to $X_t$.
                \item As a result:
                  \begin{itemize}
                    \item
                      $\Delta\OPT = 1$.
                    \item
                      $\Delta\ALG_\cB = 1$.
                    \item Note that either $b$ arrives after $d$, in which case we can add $(c, d)$ to $X_t$ and $\Delta\Phi = 0$, or the pair $(b, d)$ forms an inversion, in which case $\Delta\Phi = 1$ and $\Delta M = 1$. (As before, since $b$ is getting evicted, we will not count this pair twice.)
                  \end{itemize}
              \end{itemize}
          \end{enumerate}
      \end{enumerate}
  \end{enumerate}
It is not hard to verify that the change in the left-hand side of the bound is no more than the change in the right-hand side in each of the cases listed above, from which the proposition follows.
\end{proof}

\begin{proposition}\label{prop:one}
  The competitive ratio of algorithm $\mathcal B$ is at most $1 + 2\eps$.
\end{proposition}

\begin{proof}
  Note that $2\eta$ is bounded below by the number of inversions $M$ of $(\sigma, h)$ by \Cref{lemma:inv}. By \Cref{proposition:first}, $\mathsf{ALG}_{\mathcal A}\le\mathsf{OPT} + M$, so $\mathsf{ALG}_{\mathcal A} / \mathsf{OPT}\le 1 + M / \mathsf{OPT} \le 1 + 2\eps$.
\end{proof}

\section{A More Careful Analysis}

In this section, we give an asymptotically better (in $k$) bound for the performance of \Blind{}. A more careful analysis is needed to show an upper bound with a $1/k$ coefficient on the ratio $\eps = \eta / \mathsf{OPT}$. We use the same high-level approach for the proof as before, but with a more complicated potential function. Again, $\mathcal A$ is the offline optimal algorithm and $\mathcal B$ is the \Blind{} algorithm, and also as before, we use $\Delta$ to denote change (in functions of $t$) from request $t$ to request $t+1$.

We maintain in this proof a matching $X_t$ over pairs of page requests $(a, b)\in A_t\times B_t$ such that $h_a\ge h_b$ for each time step $t$. Our potential function $\Phi$ will be a function of $A_t$, $B_t$, and $X_t$. For notational simplicity, we will simply denote $\Phi(A_t, B_t, X_t)$ by $\Phi(t)$.

Given $A_t$, $B_t$, and $X_t$ at time $t$, define $\Phi_{0}(t)$ to be the number of $b\in B_t$ that are unmatched. Define $\Phi_1(t)$ to be the number of $b\in B_t$ such that $(b, b)\not\in X_t$. In other words, $\Phi_1$ counts how many page requests in $B_t$ are not matched to the \emph{same} page request in $A_t$. Let $z_a(t)$ be the number of pages in $B_t$ predicted to appear no later than $h_a$, with tie-breaking done in a consistent manner (e.g., by the last time the page was requested). Next, define $$
\Phi_2(t) = \sum_{(a, b)\in X_t} (z_b(t) - z_a(t)) = \sum_{(a, b)\in X_t} \paren[\big]{\phi^A_a(t) + \phi^B_b(t)},
$$ where $\phi^A_a(t) = (k - 1) - z_a(t)$ and $\phi^B_b(t) = z_b(t) - (k - 1)$. Finally, we take $$
    \Phi(t) = (k-1)\Phi_{0}(t) + (k-1)\Phi_1(t) + \Phi_2(t)
$$ as our overall potential function.

\begin{proposition}\label{proposition:second}
  For any input $(\sigma, h)$, there exists a matching $X_n\subseteq A_n\times B_n$ consisting only of pairs $(a, b)$ satisfying $h_a\ge h_b$ such that $$
    (k - 1)\mathsf{ALG}_{\mathcal B} + \Phi(n)\le 2(k-1)\mathsf{OPT} + 2M.
    $$
\end{proposition}

\begin{proof}
  We again induct on the length $n$ of the input, and we again perform a case analysis to show that we can maintain a matching $X_t$ consisting only of pairs $(a, b)$ satisfying $h_a\ge h_b$ such that at each time step, the right-hand side increases at least as much as the left-hand side.

For our analysis, we split the serving of each page request into two phases:
  \begin{enumerate}
    \item \textbf{Matching.} Update $X_t$ so that the page requests in $A_t$ and $B_t$ that are to be removed are unmatched. (Note that page requests are removed either because the corresponding page was requested again or because the corresponding page was evicted.)
    \item \textbf{Updating.} Replace a page request from each of $A_t$ and $B_t$ with the new request and insert the new page request pair $(t, t)$ into $X_t$.
  \end{enumerate}

  We first analyze how updating affects the potential $\Phi$. This operation always decreases $\Phi_{0}$ and $\Phi_1$ each by $1$, since we remove an unmatched pair. Next, for $\Phi_2$, observe that for a matched pair $(a, b)$, the difference $z_b - z_a$ increases on the $t$-th request only if there exists a $p\in B_t$ such that $\sigma_p = \sigma_t$ and $h_b < h_p \le h_a$. In this case, the pair $(p, b)$ also forms an inversion. Any inversion $(p, b)$ is counted at most once this way because $p$ is evicted from $b$. Thus, we have $\Delta\Phi_2\le\Delta M$.

We now analyze the matching phase with a case analysis:
\begin{enumerate}
\item The requested page is in both $A_t$ and $B_t$.
  \begin{itemize}
    \item The previous page requests for $\sigma_t$ in $A_t$ and $B_t$ are matched to each other, so we can just unmatch them.
    \item As a result:
      \begin{itemize}
        \item $\Delta\Phi_{0} = 1$.
        \item $\Delta\Phi_1 = 1$.
        \item $\Delta\Phi_2 = 0$, since the pages were matched to each other.
      \end{itemize}
  \end{itemize}
\item The requested page is in $A_t$ only.
  \begin{enumerate}
    \item The previous request $p\in A_t$ for the requested page is matched as $(p, d)$ and the page request $b\in B_t$ evicted by $\mathcal B$ is matched as $(c, b)$.
      \begin{itemize}
        \item Unmatch $(p, d)$ and $(c, b)$ and then match $(c, d)$. The latter is okay since $h_c\ge h_b\ge h_d$. Note that $p\neq d$.
        \item As a result:
          \begin{itemize}
            \item $\Delta\mathsf{ALG}_{\mathcal B} = 1$.
            \item $\Delta\Phi_{0} = 1$.
            \item $\Delta\Phi_1\le 1$.
            \item $\Delta\Phi_2 = z_p - (k - 1)$, since $\phi^B_b = 0$ and $\phi^A_p = (k - 1) - z_p$.
            \item $\Delta M\ge z_p$, since the arrival of $\sigma_p$ also generates $z_p$ inversions of the form $(p, b')$ for all $b'\in B_t$ such that $h_{p}\ge h_{b'}$.
          \end{itemize}
      \end{itemize}
    \item The previous request $p\in A_t$ for the requested page is matched as $(p, d)$ and the page request $b\in B_t$ evicted by $\mathcal B$ is unmatched.
      \begin{itemize}
        \item Unmatch $(p, d)$. Note that $p\neq d$.
        \item As a result:
          \begin{itemize}
            \item $\Delta\mathsf{ALG}_{\mathcal B} = 1$.
            \item $\Delta\Phi_{0} = 1$.
            \item $\Delta\Phi_1 = 0$.
            \item $\Delta\Phi_2\le z_p$, since $\phi^A_p = (k - 1) - z_p$ and $\phi^B_d = z_d - (k - 1)\ge -(k-1)$.
            \item $\Delta M\ge z_p$, since the arrival of $\sigma_p$ also generates $z_p$ inversions of the form $(p, b')$ for all $b'\in B_t$ such that $h_{p}\ge h_{b'}$.
          \end{itemize}
      \end{itemize}
    \item The previous request $p\in A_t$ for the requested page is unmatched and the page request $b\in B_t$ evicted by $\mathcal B$ is matched as $(c, b)$.
      \begin{itemize}
        \item Unmatch $(c, b)$ and match $(c, d)$ for an arbitrary unmatched $d\in B_t\setminus\{b\}$. Doing so is okay because $h_c\ge h_b\ge h_d$.
        \item As a result:
          \begin{itemize}
            \item $\Delta\mathsf{ALG}_{\mathcal B} = 1$.
            \item $\Delta\Phi_{0} = 0$.
            \item $\Delta\Phi_1\le 1$.
            \item $\Delta\Phi_2\le 0$, since $\phi^B_b = 0$ and $-\phi^B_d = (k - 1) - z_d\ge 0$.
          \end{itemize}
      \end{itemize}
    \item The previous request $p\in A_t$ for the requested page is unmatched and the page request $b\in B_t$ evicted by $\mathcal B$ is unmatched.
      \begin{itemize}
        \item Do nothing.
        \item As a result:
          \begin{itemize}
            \item $\Delta\mathsf{ALG}_{\mathcal B} = 1$.
          \end{itemize}
      \end{itemize}
  \end{enumerate}

\item The requested page is in $B_t$ only.
  \begin{enumerate}
    \item The previous request $p\in B_t$ for the requested page is matched as $(c, p)$ and the page request $a\in A_t$ evicted by $\mathcal A$ is matched as $(a, d)$.
      \begin{itemize}
        \item Unmatch $(c, p)$ and $(a, d)$. Note that $c\neq p$.
        \item As a result:
          \begin{itemize}
            \item $\Delta\mathsf{OPT} = 1$.
            \item $\Delta\Phi_{0} = 2$.
            \item If $a = d$, then $\Delta\Phi_1 = 1$ and $\phi^A_a + \phi^B_d = 0$; otherwise, $a\neq d$, in which case $\Delta\Phi_1 = 0$ and \smash{$\phi^A_a + \phi^B_d\ge -(k-1)$}.
            \item Moreover, $\phi^A_c + \phi^B_p\ge -(k-1)$.
          \end{itemize}
      \end{itemize}
    \item The previous request $p\in B_t$ for the requested page is matched as $(c, p)$ and the page request $a\in A_t$ evicted by $\mathcal A$ is unmatched.
      \begin{itemize}
        \item Unmatch $(c, p)$. Note that $c\neq p$.
        \item As a result:
          \begin{itemize}
            \item $\Delta\mathsf{OPT} = 1$.
            \item $\Delta\Phi_{0} = 1$.
            \item $\Delta\Phi_1 = 0$.
            \item $\Delta\Phi_2\le k-1$, since $\phi^A_c + \phi^B_p\ge -(k-1)$.
          \end{itemize}
      \end{itemize}
    \item The previous request $p\in B_t$ for the requested page is unmatched and the page request $a\in A_t$ evicted by $\mathcal A$ is matched as $(a, d)$.
      \begin{itemize}
        \item Unmatch $(a, d)$.
        \item As a result:
          \begin{itemize}
            \item $\Delta\mathsf{OPT} = 1$.
            \item $\Delta\Phi_{0} = 1$.
            \item If $a = d$, then $\Delta\Phi_1 = 1$ and $\phi^A_a + \phi^B_d = 0$; otherwise, $a\neq d$, in which case $\Delta\Phi_1 = 0$ and \smash{$\phi^A_a + \phi^B_d\ge -(k-1)$}.
          \end{itemize}
      \end{itemize}
    \item The previous request $p\in B_t$ for the requested page is unmatched and the page request $a\in A_t$ evicted by $\mathcal A$ is unmatched.
      \begin{itemize}
        \item Do nothing.
        \item As a result:
          \begin{itemize}
            \item $\Delta\mathsf{OPT} = 1$.
          \end{itemize}
      \end{itemize}
  \end{enumerate}
\item The requested page is in neither $A_t$ nor $B_t$.
  \begin{enumerate}
    \item The previous request $a\in A_t$ evicted by $\mathcal A$ is matched as $(a, d)$ and the page request $b\in B_t$ evicted by $\mathcal B$ is matched as $(c, b)$.
      \begin{itemize}
        \item Unmatch $(a, d)$ and $(c, b)$ and then match $(c, d)$. The latter is okay because $h_c\ge h_b\ge h_d$.
        \item As a result:
          \begin{itemize}
            \item $\Delta\mathsf{ALG}_{\mathcal B} = 1$.
            \item $\Delta\mathsf{OPT} = 1$.
            \item $\Delta\Phi_{0} = 1$.
            \item $\Delta\Phi_{1}\le 2$.
            \item $\Delta\Phi_2\le 0$, since $\phi^A_a\ge 0$ and $\phi^B_b = 0$.
          \end{itemize}
      \end{itemize}
    \item The previous request $a\in A_t$ evicted by $\mathcal A$ is matched as $(a, d)$ and the page request $b\in B_t$ evicted by $\mathcal B$ is unmatched.
      \begin{itemize}
        \item Unmatch $(a, d)$.
        \item As a result:
          \begin{itemize}
            \item $\Delta\mathsf{ALG}_{\mathcal B} = 1$.
            \item $\Delta\mathsf{OPT} = 1$.
            \item $\Delta\Phi_{0} = 1$.
            \item If $a = d$, then $\Delta\Phi_1 = 1$ and $\phi^A_a + \phi^B_d = 0$; otherwise, $a\neq d$, in which case $\Delta\Phi_1 = 0$ and \smash{$\phi^A_a + \phi^B_d\ge -(k-1)$}.
          \end{itemize}
      \end{itemize}
    \item The previous request $a\in A_t$ evicted by $\mathcal A$ is unmatched and the page request $b\in B_t$ evicted by $\mathcal B$ is matched as $(c, b)$.
      \begin{itemize}
        \item Unmatch $(c, b)$ and match $(c, d)$.
        \item As a result:
          \begin{itemize}
            \item $\Delta\mathsf{ALG}_{\mathcal B} = 1$.
            \item $\Delta\mathsf{OPT} = 1$.
            \item $\Delta\Phi_{0} = 0$.
            \item $\Delta\Phi_1\le 1$.
            \item $\Delta\Phi_2\le 0$, since $\phi^B_b = 0$ and $-\phi^B_d = (k - 1) - z_d\ge 0$.
          \end{itemize}
      \end{itemize}
    \item The previous request $a\in A_t$ evicted by $\mathcal A$ is unmatched and the page request $b\in B_t$ evicted by $\mathcal B$ is unmatched.
      \begin{itemize}
        \item Do nothing.
        \item As a result:
          \begin{itemize}
            \item $\Delta\mathsf{ALG}_{\mathcal B} = 1$.
            \item $\Delta\mathsf{OPT} = 1$.
          \end{itemize}
      \end{itemize}
  \end{enumerate}
\end{enumerate}

Since $\Phi_1$ and $\Phi_2$ each decrease by $1$ in the updating phase, we have $2(k-1)$ in extra potential that we can use to pay for costs in the matching phase. Indeed, one can verify that this is sufficient for all of the cases described above---the tight cases are 2(a), 2(b), 2(c), 3(a), and 4(a). Thus, the proposition follows.
\end{proof}

\begin{proposition}\label{prop:two}
  The competitive ratio of algorithm $\mathcal B$ is at most $2 + 4\eps / (k-1)$.
\end{proposition}

\begin{proof}
  Compose \Cref{proposition:second} with \Cref{lemma:inv}.
\end{proof}

\begin{remark*}
  This analysis of \Blind{} is tight in the constant term---we can make $\eps / (k - 1)$ arbitrarily small while having a competitive ratio of $2$. (However, for very small $\eps$, the bound of the previous section is better---in particular, if $\eps < \frac 12 + \frac 1{k-3}$.)
\end{remark*}

\section{Proofs of Upper Bounds}

We are now ready to prove the results stated in \Cref{sec:intro}:
\maintheorem*
\begin{proof}
  From the analysis of the previous two sections, the desired bound immediately follows from taking the minimum of the bounds in \Cref{prop:one,,prop:two}.
\end{proof}

\firstcorollary*
\begin{proof}
  Combine \Blind{} with \textsc{LRU} using the ``combiner'' from \Cref{theorem:det_comb}, with the performance of \Blind{} being bounded by \Cref{theorem:main}.
\end{proof}

\secondcorollary*

\begin{proof}
  Like in the proof above, combine \Blind{} with algorithm \textsc{Equitable} of \citet{ACN00}\footnote{We use \textsc{Equitable} because it achieves the optimal worst-case competitive ratio of $H_k$ for online caching; \M{} has a competitive ratio of $2H_k - 1$ \cite{ACN00}.}, this time using the ``combiner'' from \Cref{theorem:rand_comb}.
\end{proof}

\section{Deterministic Lower Bound}

We now show that combining \Blind{} with LRU gets an optimal competitive ratio bound (in terms of $\eta$, $\OPT$, and $k$) among all deterministic algorithms for learning-augmented online caching by proving \Cref{theorem:lower}:

\lowerboundtheorem*

\begin{proof}
  Let $\cA$ be any deterministic algorithm for learning-augmented online caching.

  We show there exists a family of inputs $(\sigma, h)$ with $\eps/k$ ranging from $0$ to $k$ and $\OPT$ arbitrarily large such $\ALG_\cA\ge\OPT + C\eta / k$, for some constant $C > 0$. That is, for $\eps/k$ ranging from $0$ to $k$, we will show that we can make this inequality hold as $\OPT\to\infty$ with $\eps/k$ fixed. Dividing through by $\OPT$ then implies the theorem.

  We now construct such inputs $(\sigma, h)$. First, fix $j < k$. Let $P_1,\ldots,P_k,Q_0$ be $k+1$ distinct pages. We make the following sequence of requests, which we call a \emph{phase}:
  \begin{enumerate}
    \item
      Repeat $k$ times the following:
      \begin{enumerate}
        \item
          Make requests to $P_1,\ldots,P_k$ in order, predicting each page to next appear $k$ requests from now \emph{except} during the last iteration, where we predict each page to next appear $k + j + 1$ requests from now.
      \end{enumerate}
    \item
      Make a request to $Q_0$ and predict that it will next appear $k^2 + j + 1$ pages from now.
    \item
      For $i = 1,\ldots,j$:
      \begin{enumerate}
        \item
          Request the page evicted by in $\cA$ during the previous request, if it exists. Otherwise, request an arbitrary page. For each page, provide the same prediction as the last time this page was requested.
      \end{enumerate}
  \end{enumerate}
  We repeat the above as many times as needed.

  In a single phase, observe that $\OPT$ makes at most two evictions---once to evict $Q_0$ and once upon the arrival of $Q_0$. On the other hand, I claim $\cA$ makes at least $j+1$ evictions. First, if the cache of $\cA$ after (1) does not consist of $P_1,\ldots,P_k$, then $\cA$ must have incurred cost at least $k\ge j + 1$ during (1). Thus, we may assume that $\cA$'s cache consists of $P_1,\ldots,P_k$ after (1). If so, $\cA$ has to evict a page for each of the remaining $j+1$ requests in the phase, as the arrival of $Q_0$ forces an eviction and by induction, each arrival of (3) forces an eviction. Finally, observe that all the predictions are accurate except those for pages arriving in (3), in which case they are off by at most $k + j + 1\le 2k$. Thus, over a single phase, $\eta\le 2jk$. Putting all of these observations together, we get $\ALG_\cA\ge j+1\ge\OPT + j - 1$, with $j - 1 = \Omega(\eta/k)$.

  To make $\OPT$ arbitrarily large, note that we can simply repeat the above phase multiple times in sequence; the same analysis holds, with all the values scaling linearly. Hence we have $\ALG_\cA\ge\OPT + \Omega(\eta/k)$ for arbitrarily large $\OPT$ over the desired range of $\eps/k$, so the theorem follows.
\end{proof}

\section*{Acknowledgments}

I would like to thank Jelani Nelson for advising this project and Bailey Flanigan for providing many helpful references.

\printbibliography

\end{document}